\documentclass[runningheads,11pt]{llncs}
\usepackage{fullpage,graphicx}

\newcommand{\raisecaption}{\vspace{-0.5cm}}

\newcommand{\Ob}[1]{O\!\left( #1 \right)}
\newcommand{\os}[1]{o\!\left( #1 \right)}
\newcommand{\Om}[1]{\Omega\!\left( #1 \right)}
\newcommand{\om}[1]{\omega\!\left( #1 \right)}
\newcommand{\Th}[1]{\Theta\!\left( #1 \right)}

\newcommand{\eqref}[1]{{\rm(\ref{#1})}}

\newenvironment{smallenumerate}
   {\vspace*{-4pt}
    \begin{enumerate}\itemsep=0pt}
   {\end{enumerate}
    \vspace*{-6pt}}

\newcommand{\algskip}{\itemsep=-3pt\baselineskip=11pt}
\newcommand{\eos}{\$}

\newcommand{\sa}{\mbox{{\sf sa}}}
\newcommand{\saint}{\mbox{\sf sa$_{int}$}}
\newcommand{\saext}{\mbox{\sf sa$_{ext}$}}
\newcommand{\bwtext}{\mbox{\sf bwt$_{ext}$}}
\newcommand{\bwtint}{\mbox{\sf bwt$_{int}$}}

\newcommand{\psiext}{\mbox{$\Psi_{ext}$}}

\newcommand{\pos}{\mbox{\sf pos}}

\newcommand{\bits}{\mbox{\sf gt}}

\newcommand{\sort}{\mbox{\sf sort}}

\newcommand{\bwt}{\mbox{\sf bwt}}
\newcommand{\bwtMark}{\mbox{\sf bwtMark}}
\newcommand{\gap}{\mbox{\sf gap}}

\newcommand{\A}{\Sigma}

\newcommand{\Rank}{\mbox{\sf Rank}}

\newcommand{\bwte}{\mbox{\sf bwt-disk}}
\newcommand{\dcx}{\mbox{\sf DC3}}

\newcommand{\bzip}{{\sf bzip}}
\newcommand{\ppm}{{\sf ppm}}

\newcommand{\rle}{\mbox{\sf Rle}}

\newcommand{\gzip}{\mbox{\sf gzip}}

\newcommand{\Set}{{\cal S}}

\begin{document}

\title{Lightweight Data Indexing and Compression\\in
External Memory\thanks{The first author has been partially supported by {\sf
Yahoo! Research}. The second and third authors have been supported by Italian
MIUR grant ``Italy-Israel FIRB Project". This research was done while the second author was at the Universit\`a del Piemonte Orientale, Italy.  Emails: {\tt ferragina@di.unipi.it},
{\tt manzini@mfn.unipmn.it}, {\tt travis.gagie@gmail.com}}}
\author{Paolo Ferragina\inst{1} \and Travis Gagie\inst{2} \and Giovanni Manzini\inst{3}}

\institute{$^1$\ Dipartimento di Informatica, Universit\`a di Pisa, Italy\\
$^2$\ Departamento de Ciencias de Computaci\'on, Universidad de Chile, Chile\\
$^3$\ Dipartimento di Informatica, Universit\`a del Piemonte Orientale,
Italy}

\maketitle

\begin{abstract}
In this paper we describe algorithms for computing the BWT and for building
(compressed) indexes in external memory. The innovative feature of our
algorithms is that they are lightweight in the sense that, for an input of
size $n$, they use only ${n}$ bits of disk working space while all previous
approaches use $\Th{n \log n}$ bits of disk working space. Moreover, our
algorithms access disk data only via sequential scans, thus they take full
advantage of modern disk features that make sequential disk accesses much
faster than random accesses.

We also present a scan-based algorithm for inverting the BWT that uses
$\Th{n}$ bits of working space, and a lightweight {\em internal-memory}
algorithm for computing the BWT which is the fastest in the literature when
the available working space is $\os{n}$ bits.

Finally, we prove {\em lower} bounds on the complexity of computing and
inverting the BWT via sequential scans in terms of the classic product:
internal-memory space $\times$ number of passes over the disk data.
\end{abstract}

\vspace{-.5cm}
\section{Introduction}\label{sec:intro}

\vspace{-.3cm} Full-text indexes are data structures that index a text string
$T[1,n]$ to support subsequent searches for arbitrarily long patterns like
substrings, regexp, errors, {\it etc.}, and have many applications in
computational biology and data mining. Recent years have seen a renewed
interest in these data structures since it has been proved that full-text
indexes can be compressed up to the $k$-th order empirical entropy of the
input text $T$, and searched without being fully
decompressed~\cite{NM-survey07}. At the same time, it has been shown that
modern data compressors based on full-text indexes can approach the empirical
entropy of an input string without making any assumption about its generating
source~\cite{FGMS05}. Clearly, data compression and indexing are mandatory when
the data to be processed and/or transmitted has large size.  But larger data
means more memory levels involved in their storage and hence, more costly
memory references. It is already known how to design an optimal
external-memory (uncompressed) full-text index~\cite{Fer05ssem}, and some
results on external memory compressed indexes have recently appeared in the
literature~\cite{BenderFK06,GeomBWT,GNiwoca07,SPMT08}. However, whichever is
the index chosen (compressed or uncompressed), to use it one must first {\em
build} it! The sheer size of data available nowadays for mining and search
applications has turned this into a hot topic because the
construction/compression phase may be a bottleneck that can even prevent
these indexing and compression tools from being used in large-scale
applications.

Recent research~\cite{FrMu07,Kar07smallbwt,NaPa07} has highlighted that a
major issue in the construction of such data structures is the large amount
of {\it working space} usually needed for the construction. Here working
space is defined as the space required by an algorithm in addition to the
space required for the input (the text to be indexed/compressed) and the
output (the index or the compressed file). If the data to be indexed is too
large to fit in main memory one must resort to external memory construction
algorithms. Such algorithms are known (see e.g. \cite{DKMS08,KSB06}), but
they all use $\Theta(n \log n)$ bits of working space. We found (see
Section~\ref{sec:bwt:construction}) that this working space can be up to 500
times larger than the final size of the compressed output that, for typical
data, is three to five times smaller than the original input and is anyway
$\Ob{n}$ bits in the worst case.

Given these premises, the first issue we address in this paper is the design
of construction algorithms for full-text indexes which work on a disk-memory
system and are {\em lightweight} in that their working space is as small as
possible. The second issue we address concerns the way our algorithms
fetch/write data onto disk: we design them to access disk data only via {\it
sequential} scans. This approach is motivated by the well known fact that
sequential I/Os are much faster than random I/Os. Indeed, on modern disks
sequential disk access rates are currently comparable to random access rates
in internal memory~\cite{Ruhl03}. Sequential access to data has the
additional advantage of using modern caching architectures optimally, making
the algorithm cache-oblivious. These facts are routinely exploited by expert
programmers, and have motivated a large body of research, known as {\em Data
Streaming}~\cite{Muthu-survey}. In this paper we investigate the problems of
building (compressed) full-text indexes and compressing data using only
sequential scans (i.e. streaming-like). We provide {\em upper} and {\em
lower} bounds for them in terms of the product ``internal-memory space
$\times$ passes over the disk data''.

In the following we consider the classical I/O model~\cite{Vitter01}: a fast
internal memory with $M$ words (i.e. $\Theta(M \log n)$ bits) and $O(1)$
disks of unbounded capacity. Disks are organized in pages consisting of $B$
consecutive words (i.e. $\Theta(B \log n)$ bits overall). Since our
algorithms access disk data only by sequential scans, we analyze them
counting the number of disk passes as in the streaming models: From that
number is straightforward also to derive the cost in terms of the number of
I/Os (disk page accesses).

Our first contribution is a lightweight algorithm for computing the BWT --- a
basic ingredient of both compressors and compressed indexes --- in $O(n/M)$
passes and $n$ bits of disk working space. Note that the total space usage of
the algorithm is $\Th{n}$ bits and therefore proportional to the size of the
input. Since at each pass we scan $\Th{n}$ bits of disk data, each pass scans
$\Th{n/(B \log n)}$ pages and the overall I/O complexity is $\Ob{n^2/(MB\log
n)}$. We have implemented a prototype of this algorithm (available from {\sf
http://people.unipmn.it/manzini/bwtdisk}). The prototype takes advantage of
the sequential disk access by storing all files (input, output, and
intermediate) in compressed form, thus further reducing the disk usage and
the total I/Os. Our tests show that our tool is the fastest currently
available for the computation of the BWT in external memory, and that its
disk working space is much smaller than the size of the input.

The second contribution of the paper is to show that from our algorithm we
can derive: {\bf (1)} a lightweight {\it internal-memory} algorithm for
computing the BWT, which is the fastest in the literature when the amount of
available working space is $\os{n}$ bits (Theorem~\ref{teo:bwt}), and {\bf
(2)} lightweight algorithms for computing: the suffix array, the $\Psi$
array, and a sampling of the suffix array, which are important ingredients of
(compressed) indexes (see Theorems \ref{teo:newCF}, \ref{teo:psi}, and
\ref{teo:posd}).

Another contribution is a lightweight algorithm to invert the BWT which uses
$O(n/M)$ passes with one disk or $O(\log^2 n)$ passes with two disks, and
$\Th{n}$ bits of disk working space (Theorem \ref{teo:bwtInv}). This result
is based on different techniques than the ones we used to derive our
construction algorithms.

Finally, we try to assess to what extent we can improve our scan-based
algorithms for computing/inverting the BWT with only one disk. In this
setting, lower bounds are often established considering the product
``internal-memory space $\times$ passes''~\cite{MunroP80}. For our BWT
construction and inversion algorithms such product is $\Ob{n \log n}$ bits
and, by strengthening a lower bound from one of our previous
papers~\cite{Gag09}, we prove that we cannot reduce it to $\os{n}$ bits with
a scan-based algorithm using a single disk (Theorem \ref{teo:LBs}). Hence our
algorithms are within an $O(\log n)$ factor of the optimal.  We note that our
lower bound is ``best possible'' because, if we have $\Omega (n)$ bits of
memory, then we can read the input into internal memory with one pass over
the disk and then compute the BWT there.

\bigbreak\noindent {\bf Related results.} As we mentioned above, the problem
of the lightweight computation of (compressed) indexes in internal memory has
recently received much attention
(see~\cite{FrMu07,Hon07,Kar07smallbwt,NaPa07,Siren09} and references
therein). However, all the proposed algorithms perform many random
memory-accesses so they cannot be easily transformed into external memory
algorithms. To our knowledge no lightweight algorithms specific for external
memory are known. The construction of most full-text indexes reduces to
suffix-array construction, which in turn needs $\log n$ recursive
sorting-levels~\cite{FFM:00}. In external memory this sort-based approach
takes $O\bigl({\frac{n}{B}\log_{M/B} \frac{n}{B}}\bigr)$
I/Os~\cite{Fer05ssem} and is faster than our algorithms when $M = O
\bigl({{n}/\bigl({\log n \log_{M/B}\frac{n}{B}}\bigr)}\bigr)$. However, the
sort-based approach is not lightweight since it uses $\Th{n\log n}$ bits of
disk working space.

\vspace{-.2cm}

\section{Notation}\label{sec:notation}

\vspace{-.2cm} We briefly recall some definitions related to compressed
full-text indexes; for further details see~\cite{NM-survey07}. Let $T[1,n]$
denote a text drawn from a constant size alphabet~$\A$. As is usual, we
assume that $T[n]$ is a character not appearing elsewhere in $T$ and is
lexicographically smaller than all other characters. Given two strings $s$,
$t$ we write $s\prec t$ to denote that $s$ precedes $t$ lexicographically.
The suffix array $\sa[1,n]$ is the permutation of $[1,n]$ giving the
lexicographic order of the suffixes of $T$, that is $T[\sa[i],n] \prec
T[\sa[i+1],n]$ for $i=1,\ldots,n-1$. The inverse of the \sa\ is the \pos\
array, such that $\pos[i]$ is the rank of suffix $T[i,n]$ in the suffix
array. This way, $\sa[\pos[i]] = i$. We denote by $\pos_d$ the set of $(n/d)$
values $\pos[d], \pos[2d], \ldots, \pos[n]$ that indicate the distribution of
the positions of the $d$-spaced suffixes within $\sa$.

The Burrows-Wheeler transform is an array of characters $\bwt[1,n]$ defined
as $\bwt[i] = T[(\sa[i]-1) \bmod n]$. The array $\Psi[1,n]$ is the
permutation of $[1,n]$ such that $\sa[\Psi(i)] = \sa[i] + 1 \bmod n$. The
value $\Psi[i]$ is the lexicographic rank of the suffix which is one
character shorter than the suffix of rank~$i$. The basic ingredients of most
compressed indexes are either the \bwt\ or the $\Psi$ array, optionally
combined with the set $\pos_d$ for some $d=\Om{\log n}$. In this paper we
describe external memory lightweight algorithms for the computation of all
these three basic ingredients.

\vspace{-.2cm}

\section{Lightweight Scan-Based BWT construction}
\label{sec:bwt:construction}

\vspace{-.2cm} In this section we describe the algorithm \bwte\ for the
computation of the \bwt\ of a text $T[1,n]$ when $n$ is so large that the
computation cannot be done in internal memory. Our algorithm is lightweight
in the sense that it uses only $M$ words of RAM and $n$ bits of disk space
--- in addition to the disk space used for the input $T[1,n]$ and the output
$\bwt(T[1,n])$. Our algorithm is scan-based in the sense that all data on
disk is accessed by sequential scans only. Note that in the description below
our algorithm scans the input file right-to-left: in the actual
implementation we scan the input rightward which means that we compute the
\bwt\ of $T$ reversed. The
\bwte\ algorithm is an evolution of a disk-based construction algorithm for
suffix arrays first proposed in~\cite{pat-tree} and improved
in~\cite{CraFer02}. However, our algorithm constructs the \bwt\ {\em
directly} without passing through the \sa\ and uses some new ideas to reduce
the working space from $\Theta(n \log n)$ to $n$ bits.

The algorithm \bwte\ logically partitions the input text $T[1,n]$ into blocks
of size $m = \Theta(M)$ characters each, i.e. $T = T_{n/m} T_{n/m -1} \cdots
T_2 T_1$, and computes incrementally the \bwt\ of $T$ via $n/m$ passes, one
per block of~$T$. Text blocks are examined right to left so that at pass
$h+1$ we compute and store on disk $\bwt(T_{h+1} \cdots T_1)$ given $\bwt(T_h
\cdots T_1)$. The fundamental observation is that going from $\bwt(T_h \cdots
T_1)$ to  $\bwt(T_{h+1} \cdots T_1)$ requires only that we insert the
characters of $T_{h+1}$ in $\bwt(T_h \cdots T_1)$. In other words, adding
$T_{h+1}$ does not modify the relative order of the characters already in
$\bwt(T_h \cdots T_1)$.

\begin{figure}[t]
\hrule\small{\null\algskip} At the beginning of pass $h+1$, we
assume that $\bwtext$ contains the \bwt\ of $T_{h} T_{h-1} \cdots
T_1$, and the bit array $\bits$ is defined as described in the text.
Both arrays are stored on disk.
\vspace{2mm}
\begin{smallenumerate}

\item Compute in internal memory the array $\saint[1,m]$ which contains the
    lexicographic ordering of the suffixes starting in $T_{h+1}$ and
    extending up to $T[n]$ (the end of $T$). This step uses $T_{h+1}, T_{h}$
    and the first $m-1$ entries of $\bits$.\label{step:saint} Let us call the
    suffixes starting in $T_{h+1}$ {\it new} suffixes, and the ones starting
    in $T_{h} \cdots T_1$ {\it old} suffixes.

\item Compute in internal memory the array $\bwtint[1,m]$ defined as
    $\bwtint[i] = T_{h+1}[\saint[i]-1]$, for $i=1,\ldots, m$. If
    $\saint[i]=1$ set $\bwtint[i] = \#$ where \# is a character not appearing
    in $T$.\label{step:bwt}

\item Using $\bwtint$ and scanning both $T_{h} T_{h-1} \cdots T_1$ and
    $\bits$, compute how many old suffixes fall between two lexicographically
    consecutive new suffixes. At the same time update \bits\ so that it contains
    the correct information for the extended string $T_{h+1} T_{h} \cdots
    T_1$.  \label{step:between}

\item Merge $\bwtext$ and $\bwtint$ so that at the end of the step $\bwtext$
    contains the \bwt\ of $T_{h+1} T_{h} \cdots T_1$.\label{step:update}

\end{smallenumerate}
\vspace{2mm}\hrule\caption{Pass $h+1$ of the \bwte\ algorithm to compute
$\bwt(T_{h+1}\cdots T_1)$ given $\bwt(T_{h}\cdots T_1)$.}\label{fig:sa}
\end{figure}

At the beginning of pass $h+1$, in addition to the \bwt\ of $T_h \cdots T_1$
we assume we have on disk a bit array, called $\bits$, such that $\bits[i]=1$
if and only if the suffix $T[i,n]$ starting in $T_h \cdots T_1$ is {\sf
g}reater {\sf t}han the suffix $T_h \cdots T_1$ (hence at pass $h+1$ this
array takes exactly $hm-1$ bits). For simplicity of exposition, we denote by
$\bits_h[1,m-1]$ the part of the array $\bits$ referring to the text suffixes
which start in $T_h$: namely, it is $\bits_h[i]=1$ iff the suffix starting at
$T_{h}[1+i]$ is lexicographically greater than the suffix starting at
$T_{h}[1]$, for $i=1,\ldots,m-1$ (note that all these suffixes extend past
$T_h$ up to the last character of $T$).

The pseudo-code of the generic $(h+1)$-th pass is given in
Figure~\ref{fig:sa}. Step \ref{step:saint} reads into internal memory the
substring $t[1,2m] = T_{h+1} T_{h}$ and the binary array $\bits_h[1,m-1]$.
Then we build \saint\ by lexicographically sorting the suffixes starting in
$T_{h+1}$ and possibly extending up to $T[n]$ (the last character of $T$).
Observe that, given two such suffixes starting at positions $i$ and $j$ of
$T_{h+1}$, with $i<j$, we can compare them lexicographically by comparing the
strings $t[i,m]$ and $t[j,j+m-i]$, which have the same length and are
completely contained in $t[1,2m]$ (thus, they are in internal memory). If
these strings differ we are done; otherwise, the order between the above two
suffixes is determined by the order of the suffixes starting at $t[m+1]
\equiv T_h[1]$ and $t[j+m-i+1] \equiv T_h[1+j-i]$. This order is given by the
bit stored in $\bits_h[j-i]$, also available in internal memory. This
argument shows that $t[1,2m]$ and $\bits_h$ contain all the information we
need to build \saint\ working in internal memory. The actual computation of
\saint\ is done in $\Ob{m}$ time as follows. First we compute the rank
$r_{m+1}$ of the suffix starting at $t[m+1]\equiv T_h[1]$ among all suffixes
starting in $T_{h+1}$; that is, we compute for how many indices $i$ with
$1\leq i \leq m$ the suffix starting at $t[i]$ is smaller than the suffix
starting at $t[m+1]$ (both extending up to $T[n]$). This can be done in
$\Ob{m}$ time using the above observation and Lemma~5
in~\cite{Kar07smallbwt}. At this point the problem of building \saint\ is
equivalent to the problem of building the suffix array of the string
$t[1,m]\eos$, where $\eos$ is a special end-of-string character that has rank
precisely $r_{m+1}$ (instead of being lexicographically smaller than all
other suffixes, as is usually assumed). Thus, we can compute \saint\ in
$\Ob{m}$ time and $\Ob{m \log m}$ bits of space with a straightforward
modification of the algorithm DC3~\cite{KSB06}.

At Step~\ref{step:bwt} we build the array \bwtint\ which is a sort of \bwt\
of the string $T_{h+1}$: it is not a {\em real} \bwt\ because it refers to
suffixes which are not confined to $T_{h+1}$ but start in this string and
extend up to $T[n]$. The crucial point of the algorithm is then to compute
some additional information that allows us to {\em merge} \bwtint\ and
\bwtext\ I/O-efficiently. This additional information consists of a counter
array $\gap[0,m]$ which stores in $\gap[j]$ the number of (old) suffixes of
the string $T_h \cdots T_1$ which lie lexicographically between the two new
suffixes--- i.e. $\saint[j-1]$ and $\saint[j]$--- starting in $T_{h+1}$. Note
that the \gap\ array was used also in~\cite{CraFer02}. However
in~\cite{CraFer02} \gap\ is computed in $O(n\log M)$ time using $\Th{n\log
n}$ extra bits; here we compute \gap\ in $O(n)$ time using only the $n$ extra
bits of \bits. The following lemma is the key to this improvement.

\begin{lemma}\label{lemma:bits}
For any character $c\in\Sigma$, let $C[c]$ denote the number of characters in
\bwtint\ that are smaller than $c$, and let $\Rank(c,i)$ denote the number of
occurrences of $c$ in the prefix $\bwtint[1,i]$. Assume that the old suffix
$T[k,n]$ is lexicographically larger than precisely $i$ new suffixes, that
is,
$$
T[\saint[i],n] \prec T[k,n] \prec T[\saint[i+1],n].
$$
Now fix $c=T[k-1]$. Then, the old suffix $T[k-1,n]=c \, T[k,n]$ is
lexicographically larger than precisely $j$ new suffixes, that is,
$T[\saint[j],n] \prec T[k-1,n] \prec T[\saint[j+1],n]$, where
$$
j = \cases{
 C[c] + \Rank(c,i) & \mbox{if $c\neq T_{h+1}[m]$;}\cr
 C[c] + \Rank(c,i) + \bits[k] & \mbox{if $c=T_{h+1}[m]$.}}
$$
\end{lemma}

\begin{proof}
Obviously $T[k-1,n]$
is larger than the new suffixes that start with a character smaller than $c$
(they are $C[c]$), and is smaller than all new suffixes starting with a
character greater than $c$. The crucial point is now to compute how many new
suffixes $T[\ell,n]$ starting with $c=T[k-1]$ are smaller than $T[k-1,n]$.
(Recall that $T[\ell,n]$ starts in $T_{h+1}$, and $T[k,n]$ starts in $T_h
\cdots T_1$.)

Consider first the case $c\neq T_{h+1}[m]$.  Since $T[\ell]=c\neq
T_{h+1}[m]$, we have that $T[\ell+1,n]$ is also a new suffix (i.e. it lies in
$T_{h+1}$) and $T[\ell,n] \prec T[k-1,n]$ iff $T[\ell+1,n] \prec T[k,n]$.
Furthermore, $T[\ell]=T[k-1]=c$ so that the sorting of the rows in BWT
implies that, counting how many new suffixes starting with $c$ are smaller
than $T[k-1,n]$ {\em is equivalent to} counting how many $c$'s occur in
$\bwtint[1,i]$. This is precisely $\Rank(T[k-1],i)$. Assume now that
$c=T_{h+1}[m]$. Among the new suffixes starting with $c$ there is also the
one starting at position $T_{h+1}[m]$, call it $T[\ell',n]$. We cannot use
the above trick to compare $T[k-1,n]$ with $T[\ell',n]$ since $T[\ell'+1,n]$
coincides with $T_h \cdots T_1$ and is therefore an old suffix, not a new one
and thus not occurring in \saint. However, it is still true that $T[\ell',n]
\prec T[k-1,n]$ iff $T[\ell'+1,n]\prec T[k,n]$ and since $T[\ell'+1,n]=T_h
\cdots T_1$ we know that this holds iff $\bits[k]=1$.\qed
\end{proof}

Step~\ref{step:between} uses the above lemma to compute the array $\gap$ with
a single right-to-left scan of the two arrays $T_h \cdots T_1$ and $\bits$
available on disk. Step~\ref{step:between} takes $\Ob{n}$ time because we can
build a $o(m)$-bit data structure supporting $O(1)$ time \Rank\ queries over
\bwtint~\cite{NM-survey07}. Finally, Step~\ref{step:update} uses \gap\ to
create the new array $\bwtext$ by merging $\bwtint$ with the current
$\bwtext$. The idea is very simple: for $i=0,\ldots,m-1$ we copy $\gap[i]$
old values in $\bwtext$ followed by the value $\bwtint[i+1]$.

Note that at Step~\ref{step:between} we also compute the content of \bits\
for the next pass: namely, $\bits[k]=1$ iff $T_{h+1} \cdots T_1 \prec
T[k,n]$. We know the lexicographic relation between $T_{h+1} \cdots T_1$ and
all new suffixes since it does exist $r_1$ such that $T[\saint[r_1],n] =
T_{h+1} \cdots T_1$ (the latter is a new suffix, indeed). The relation
between $T_{h+1} \cdots T_1$ and any old suffix $T[k,n]$ is available during
the construction of $\gap$: when we find that $T[k,n]$ is larger than $i$ new
suffixes of $\saint$, we know that $T_{h+1} \cdots T_1 \prec T[k,n]$ iff $r_1
\leq i$. So we can write the correct value for $\bits[k]$ to disk.

It is easy to see that our algorithm uses $O(m\log m)$ bits of internal
memory. Hence, if the internal memory consists of $M$ words, we can take
$m=\Th{M}$ and establish the following result.

\begin{theorem}\label{teo:bwte}
We can compute the \bwt\ of a text $T[1,n]$ in $\Ob{n/M}$ passes over
$\Th{n}$ bits of disk data, using $n$ bits of disk working space. The total
number of I/Os is $\Ob{n^2/(MB\log n)}$ and the CPU time is $\Ob{n^2/M}$.
\qed
\end{theorem}

\medbreak\noindent{\bf Single-disk implementation.} In the \bwte\ algorithm,
and in its derivatives described below, we scan $T$ and the \bits\ array in
parallel so we need at least {\em two disks}. However, in view of the lower
bounds in Section~\ref{sec:lb}, which hold for a single disk, it is important
to point out that our algorithm (and its derivatives) can work via sequential
scans using {\em only one} disk. This is possible by interleaving $T$ and the
\bits\ array in a single file. At pass $h$ we interleave $m$ new bits within
the segment $T_h$ (so that the portion $T_{n/m} \cdots T_{h+1}$ is shifted by
$m$ bits). These new bits together with the bits already interleaved in
$T_{h-1} \cdots T_1$ allow us to store the portion of the \bits\ array that
is needed at the next pass. Note also that the merging of \bwtext\ and
\bwtint\ at Step~\ref{step:update} can be done on a single disk. This
requires that, at the beginning of the algorithm, we reserve on disk the
space for the full output ($n$ characters), and that we fill this space
right-to-left (that is, at the end of pass $h$ $\bwt(T_h \cdots T_1)$ is
stored in the rightmost $mh$ characters of the reserved space).

\medbreak\noindent{\bf Working with compressed files.} Accessing files only
by sequential scans makes it possible to store them on disk in compressed
form. This is not particularly significant from a theoretical point of view
--- in the worst case the compressed files still take $\Th{n}$ bits --- but is a
significant advantage in practice. If the input file $T[1,n]$ is large, it is
likely that it will be given to us in compressed form. If the compression
format allows for the scanning of a file without full decompression (as, for
example, \gzip, \bzip, and \ppm) our algorithm is able to work on the
compressed input without additional overhead. An algorithm that accesses the
input non-sequentially would require the additional space for an uncompressed
image of $T[1,n]$. The same considerations apply to the output file $\bwt(T)$
and the intermediate files $\bwt(T_h \cdots T_1)$. Since they are \bwt's of
(suffixes of) $T$ they are likely to be highly compressible, so it is very
convenient to be able to store them in compressed form: this makes our
algorithm even more ``lightweight''. It goes without saying that using
compressed files also yields a reduction of the I/O transfer so this is
advantageous also in terms of running time (see experimental results below).

Note that the use of compressed files is straightforward if we use two disks:
in this way we can store $T$ and \bits\ separately and at Step~4 we can store
on two different disks the compressed images of $\bwt(T_h \cdots T_1)$ and
$\bwt(T_{h+1} \cdots T_1)$. The use of compression in the single disk version
is trickier and requires the use of ad-hoc compressors.

\medbreak\noindent{\bf Experimental results.} To test how \bwte\ works in
practice, we have implemented a prototype in~{\sf C} {(source code available
at the page {\sf http://people.unipmn.it/manzini/bwtdisk})}. The main
modification wrt the description of Fig.~\ref{fig:sa} is that, instead of
storing the entire array \bits\ on disk, we maintain a ``reduced'' version in
RAM. In fact, Step~\ref{step:saint} uses $\bits_h[1,m-1]$ which can be stored
in RAM. At Step~\ref{step:between} we need the entire \bits\ to
lexicographically compare all suffixes $T[k,n]$ of $T_h \cdots T_1$ with $T_h
\cdots T_1$ itself (see proof of Lemma~\ref{lemma:bits}). Instead of storing
the whole \bits, we keep in internal memory the length-$\ell$ prefix of $T_h
\cdots T_1$, call it $\alpha_h$, and the entries $\bits[k]$ such that
$\alpha_h$ is a prefix of $T[k,n]$. Unless $T$ is a very pathological string,
this ``reduced'' version of \bits\ is much more succinct: by setting
$\ell=1024$ we were able to store it in internal memory in just 128KB. Using
this ``reduced'' version, the comparison between $T[k,n]$ and $T_h \cdots
T_1$, can be done by comparing $T[k,n]$ with $\alpha_h$. If these two strings
are different, we are done; otherwise, $\alpha_h$ is a prefix of $T[k,n]$ and
thus the bit $\bits[k]$ is available and provides the result of that suffix
comparison. Hence, by using standard string-matching techniques, it is
possible to compare all suffixes $T[k,n]$ with $T_h \cdots T_1$ in
$\Ob{n+\ell}$ time overall.

Our implementation can work with a block size $m$ of up to 4GB and uses $8m$
bytes of internal memory for the storage (and computation) of \saint,
\bwtint, and the \gap\ array. We ran our experiments with $m=400$MB
on a Linux box with a 2.5Ghz AMD Phenom 9850 Quad Core processor (only one
CPU was used for our tests) and 3.7GB of RAM.  On the same machine we also
tested the best competitor of our algorithm. Since all other known approaches
for computing the BWT in external memory compute the suffix array first, we
tested the \dcx\ tool~\cite{DKMS08} which is the current best algorithm for
computing the suffix array in external memory. We ran \dcx\ using two disks
for the storage of temporary files and setting the {\sf ram\_usage} parameter
to 1500MB. With these settings the peak heap memory usage reported by {\sf
memusage} was between 3.2 and 3.3 Gigabytes for both \bwte\ and \dcx.

\begin{figure}[t]

\def\Dsize{11cm}

\begin{center}
{\footnotesize

\begin{tabular}{|l|r|l|}

\hline

File Name & Size (GB) & Description \\ \hline

{\sf Proteins} & 1.10 &
\begin{minipage}[t]{\Dsize}
Sequence of bare protein sequences from the Pizza\&Chili
corpus~\cite{pizzachili_home}.
\end{minipage}\\\hline

{\sf Swissprot} & 1.88 &
\begin{minipage}[t]{\Dsize}
Annotated Swiss-prot Protein knowledge base (file {\sf uniprot\_sprot.dat}
downloaded from {\sf ftp://ftp.ebi.ac.uk/pub/databases/swissprot/release} on
June 2009).
\end{minipage}\\\hline

{\sf Genome} & 2.86 &
\begin{minipage}[t]{\Dsize}
Human genome (May 2004 version) filtered in order to have a string over the
alphabet A,C,G,T,N. This is the same file used in~\cite{DKMS08}.
\end{minipage}\\\hline

{\sf Gutenberg} & 3.05 &
\begin{minipage}[t]{\Dsize}
Concatenation of English texts from Project Gutenberg. This is the same file
used in~\cite{DKMS08}.
\end{minipage}\\\hline

{\sf Random2} & 4.00 &
\begin{minipage}[t]{\Dsize}
Concatenation of two copies of a string of length $2$GB with characters
randomly generated over an alphabet of size 128 using the tool {\sf
gentext}~\cite{pizzachili_home}.
\end{minipage}\\\hline

{\sf Mice\&Men} & 5.43 &
\begin{minipage}[t]{\Dsize}
Concatenation of the Mouse (mm9) and Human (hg18) genomes
filtered in order to have a string over the alphabet A,C,G,T,N.
\end{minipage}\\\hline

{\sf Html} & 8.00 &
\begin{minipage}[t]{\Dsize}
First $8$ GB of file {\sf Law03} from~\cite{dataset} consisting of a
collection of {\sf html} pages crawled from the {\sf UK}-domain in 2006-07
(with the WARC headers removed).
\end{minipage}\\\hline

\end{tabular}

\bigbreak

\begin{tabular}{|l|r||r|r||r|r|r|r||r|r|}\hline

\multicolumn{2}{|c||}{Input file}  & \multicolumn{2}{c||}{\bwte\
space}&\multicolumn{4}{c||}{\bwte\ time}&
\multicolumn{2}{c|}{\dcx}\\
\hline

{name\hspace{1.3cm}} & {size\ \ \ \ \ } & {\it output/working} & {\it total\
\ \ \ } &
 {\it step 1} & {\it step 3} & {\it step 4}&{\it total\ \ \ } &  time\ \ \ \ & work space
\\\hline\hline

{\sf Proteins} &1.10&	0.29&	1.02&	0.49&	0.59& 0.15& 1.45& 6.31& 30.62\\\hline
{\sf SwissProt}&1.88&	0.08&	0.33&	0.36&	1.16& 0.10& 1.88& 6.67& 30.68\\\hline
{\sf Genome}   &2.86&	0.22&	0.69&	0.50&	2.32& 0.55& 3.72& 6.88& 30.68\\\hline
{\sf Gutenberg}&3.05&	0.18&	0.74&	0.85&	2.19& 0.36& 3.76& 7.14& 30.58\\\hline
{\sf Random2}  &4.00&	0.56&	2.00&	0.80&	3.28& 2.32& 6.90& 7.48& 30.66\\\hline
{\sf Mice\&Men-4}&4.00&	0.22&	0.70&	0.52&	3.37& 0.80& 5.06& 7.22& 30.66\\\hline
{\sf Html-4}   &4.00&	0.06&	0.33&	0.58&	2.55& 0.19& 3.60& 7.49& 30.66\\\hline
{\sf Mice\&Men}&5.43&	0.22&	0.70&	0.52&	4.06& 1.14& 6.10&\multicolumn{1}{c|}{---}& \multicolumn{1}{c|}{---}\\\hline
{\sf Html}     &8.00&	0.05&	0.32&	0.49&	4.90& 0.33& 6.01&\multicolumn{1}{c|}{---}& \multicolumn{1}{c|}{---}\\\hline

\hline

\end{tabular}

}
\end{center}

\raisecaption \caption{\footnotesize\label{fig:experiments} Dataset (top) and
experimental results (bottom). Since \dcx\ cannot handle files larger than
4GB we considered also the files {\sf Mice\&Men} and {\sf Html} truncated at
4GB (indicated by the suffix {\sf -4}). In the bottom table, column 2 reports
the size (in Gigabytes) of the uncompressed input file: the values in all
other columns are normalized with respect to this size. Column~3 reports the
size of the compressed \bwt\ which is also an upper bound to the working
space of \bwte (see text). Column~4 reports the total (working $+$ input $+$
output) disk space used by \bwte. Columns 5--9 report running (wallclock)
times in microseconds per input byte. The last column reports the size of
\dcx\ working space (again normalized with respect to the size of the input
file).}

\end{figure}

In our implementation we store the files in compressed form: the input $T$ is
\gzip-compressed, whereas the partial (and final) \bwt's are compressed by
\rle\ followed by {\em range coding}: according to the experiments
in~\cite{FGM06b} this combination offers the best compression/speed tradeoff
for compressing the BWT. Our current implementation uses a single disk. Since
at Step~\ref{step:update} we scan simultaneously two partial \bwt's (say
$\bwt(T_h \cdots T_1)$ and $\bwt(T_{h+1}\cdots T_1)$) in that step the disk
head has to move between the two files and the algorithm is not
``scan-only''. We plan to support the use of two disks to remove this
inefficiency in a future version.

Our algorithm stores on disk only the compressed input and at most two
compressed partial \bwt's. Hence, the working space (the space used in
addition to the input and the output) has the size of a single compressed
partial~\bwt: in Fig.~\ref{fig:experiments} we bound it with the size of the
final compressed \bwt. The results in Fig.~\ref{fig:experiments} show that
our algorithm is indeed lightweight: for all files the working space is
(much) smaller than the size of the input text uncompressed; for most files
even the {\it total} space usage is less than the size of the uncompressed
input. The algorithm \dcx\ uses consistently a working space of more than
$30$ times the size of the uncompressed input. Comparing columns 3 and 9 we
see that, for all files except {\sf Random2}, \dcx\ working space is more
than 100 times the size of the compressed \bwt; for the file {\sf Html-4},
which is highly compressible, \dcx\ working space is more than 500 times the
size of the compressed \bwt! (recall that \bwte\ working space is at most the
size of the compressed \bwt).

By comparing the running times (columns 8 and 9 in
Fig.~\ref{fig:experiments}) we see that \bwte\ is always faster than \dcx\
(recall that \dcx\ only computes the suffix array so we are ignoring the
additional cost of computing the BWT from the suffix array). The results show
that the more compressible is the input, the faster is \bwte, while \dcx's
running time is much less sensitive to the content of the input file. Another
interesting data is the total I/O volume of the two algorithms (measured as
the ratio between total I/Os and input size and not reported in
Fig.~\ref{fig:experiments}). According to~\cite{DKMS08} for files up to 4GB
for \dcx\ such ratio is between 200 and 300. For \bwte\ such ratio is less
than 6 for all files except {\sf Random2} for which the ratio is 14.76.

The asymptotic analysis predicts that, if $M\ll n$, as the size of the input
grows, our algorithm will eventually become slower than \dcx\ (our algorithm
is designed to be lightweight, not to be fast!). However, the above results
show that the use of compressed files and avoiding the construction of the
suffix array make our algorithm, not only lightweight, but also faster than
the available alternatives on real world inputs.

\vspace{-.2cm}

\section{Other Lightweight Scan-Based Construction
Algorithms}\label{sec:other}

\vspace{-.2cm}\noindent{\bf Internal Memory Lightweight BWT construction.}
Our \bwte\ algorithm can be turned into a lightweight {\it internal memory}
algorithm with interesting time-space tradeoffs. For example, setting
$M=n/\log n$ we get an internal memory algorithm that runs in $\Ob{n\log n}$
time and uses $2{n}$ bits of working space: $n$ bits for the \bits\ array and
$n$ bits for the $M$ words that play the same role as the internal memory in
\bwte. Setting $M=n/\log^{1+\epsilon} n$, with $\epsilon>0$, the running time
becomes $\Ob{n\log^{1+\epsilon} n}$ and the working space is reduced to
$n+o(n)$ bits. This algorithm still accesses the text and the partial \bwt's
by sequential scans, hence it takes full advantage of the very fast caches
available on modern CPU's.

We can further reduce the working space by replacing the $n$ bits of the
\bits\ array with a $o(n)$-bit data structure supporting $O(1)$-time \Rank\
queries over $\bwt(T_h \cdots T_1)$. This data structure can provide in
constant time the lexicographic rank of each suffix of $T_h \cdots T_1$ (in
right-to-left order, see~\cite{NM-survey07}) and therefore can emulate,
without asymptotic slowdown, the scanning of~\bits.

{If we no longer need the input text $T$, we can write the (partial) \bwt's
over the already processed portion of text. That is, at the end of pass $h$,
we store $\bwt(T_h \cdots T_1)$ in the space originally used for $T_h \cdots
T_1$. The right-to-left scan of $T_h \cdots T_1$ required at
Step~\ref{step:between} can be emulated, without any asymptotic slowdown,
using the same data structure used to replace \bits\ (see
again~\cite{NM-survey07}). Note that overwriting $T$ roughly doubles the size
of the largest input that can be processed with a given amount of internal
memory.} Summing up, we have:

\begin{theorem}\label{teo:bwt}
For any $\epsilon > 0$, we can compute the BWT in internal memory in $\Ob{n
\log^{1+\epsilon} n}$ time, using $o(n)$ bits of working space.  The BWT can be stored in the space originally containing the input text. \qed
\end{theorem}

The only internal-memory BWT construction algorithm that can use such a small
working space is~\cite{Kar07smallbwt} which---when restricted to using
$\os{n}$ bits of working space---runs in $\om{n\log^2 n}$ time. Note,
however, that the algorithm~\cite{Kar07smallbwt} has the advantage of working
also for non constant alphabets and can use as little as $\Th{n\log
n/\sqrt{v}}$ bits of working space with $v = \Ob{n^{2/3}}$, running in
$\Ob{n\log n + vn}$ worst case time.

The algorithms in~\cite{Hon07,NaPa07} build directly a compressed suffix
array but, at least in their original formulation, they use $\Om{n}$ bits of
working space. The algorithm in~\cite{Siren09} build  a compressed suffix
array of a collections of texts. Note that building a (compressed) suffix
array for a collection of texts of total length $n$ is a different (simpler)
problem than building a (compressed) suffix array of single text of length
$n$: in a collection each text is terminated by a unique {\sf eof} symbol so
there cannot be very long common prefixes. For a collection of $p=\Th{\log
n}$ texts of size $n/p$ the algorithm in~\cite{Siren09} runs in $\Ob{n\log
n}$ time using $O(n)$ bits of working space, storing the output in compressed
form and overwriting the input. The algorithm in~\cite{Siren09} is based on
the same techniques from~\cite{pat-tree} that we use in this
paper.\footnote{Note that \cite{Siren09} was published after the first draft
of this paper~\cite{bwtext_tr} was completed.} However we merge $\bwt(T_h
\cdots T_1)$ and $\bwt(T_{h+1})$ following the original idea
of~\cite{pat-tree} of locating the suffixes of $T_h \cdots T_1$ inside the
(compressed) suffix array of $T_{h+1}$, while \cite{Siren09} does the
opposite. This choice implies other differences: for example, instead of the
\gap\ array \cite{Siren09} builds an array of ranks which is later sorted to
perform the merging.

\bigbreak\noindent{\bf Lightweight SA construction.} We can transform our
\bwte\ algorithm into a lightweight algorithm for computing the Suffix Array.
The key observation is that the values stored in $\bwtext$ are never used in
subsequent computations. Therefore, to compute the \sa, we can simply replace
$\bwtext$ with an array $\saext$ containing the \sa\ entries (that is, at the
end of pass~$h$ $\saext$ contains $\sa(T_h \cdots T_1)$). The only change in
the algorithm is that, after the computation of the \gap\ array, at
Step~\ref{step:update} we update \saext\ as follows: we copy $\gap[i]$ old
\saext\ entries followed by $\saint[i+1]$, for $i=0,\ldots,m-1$. Summing up,
we have the following result.

\begin{theorem}\label{teo:newCF}
We can compute the suffix array in $\Ob{n/M}$ passes over $\Th{n\log n}$ bits
of disk data, using $n$ bits of disk working space. The total number of I/Os
is $\Ob{n^2/(MB)}$ and the CPU time is $\Ob{n^2/M}$.\qed
\end{theorem}

Note that compared to the algorithm in~\cite{CraFer02}, which has a similar
structure and similar features, our new proposal reduces the working space
(and thus the amount of processed data), and the CPU time, by a logarithmic
factor.

\bigbreak\noindent{\bf Lightweight Computation of the $\Psi$ Array.} We use
the same framework as above and maintain an array $\psiext$ that, at the end
of pass $h$, contains the $\Psi$ values for the string $T_{h} \cdots T_1$.
Since the value $\Psi[j]$ refers to the suffix of lexicographic rank $j$, at
Step~\ref{step:update} $\psiext$ values are computed using the same scheme
used for BWT and suffix array entries: for $i=0,\ldots,m-1$, we first update
$\gap[i]$ values in $\psiext$ referring to old suffixes and then compute and
write the $\Psi$ value referring to $T[\saint[i+1],n]$.  We can compute
$\Psi$ values for the new suffixes using information available in internal
memory, while for old suffixes we make use of the relationship $\Psi_{h+1}[j]
= \Psi_{h}[j] + k_j$ where $k_j$ is the largest integer such that $\gap[0] +
\gap[1] + \cdots + \gap[k_j] < \Psi_{h}[j]$ (details in the full paper).
Since each value $k_j$ can be computed in $\Ob{\log m}$ time with a binary
search over the array whose $i$-th element is $\gap[0] + \cdots + \gap[i]$,
we have the following result.

\begin{lemma}\label{lemma:psi}
We can compute the array $\Psi$ in $\Ob{n/M}$ passes over $\Th{n\log n}$ bits
of disk data, using $n$ bits of working space. The CPU time is $\Ob{(n^2\log
M)/M}$.\qed
\end{lemma}

To reduce the amount of processed data, we observe that although $\Psi$
values are in the range $[1,n]$, it is well known~\cite{NM-survey07} that the
sequence $\Psi[1], \Psi[2] - \Psi[1], \Psi[3] - \Psi[2], \ldots, \Psi[n] -
\Psi[n-1]$, can be represented in $\Th{n}$ bits. Thus, by storing an
appropriate encoding of the differences $\Psi[i] - \Psi[i-1]$ we can obtain
an algorithm that works over a total of $\Ob{n}$ bits.

\begin{theorem}\label{teo:psi}
We can compute the array $\Psi$ in $\Ob{n/M}$ passes over $\Th{n}$ bits of
disk data, using ${n}$ bits of disk working space. The total number of I/Os
is $\Ob{n^2/(MB\log n)}$ and the CPU time is $\Ob{(n^2\log M)/M}$.\qed
\end{theorem}

\bigbreak\noindent{\bf Lightweight Computation of $\pos_d$.} To compute the
set $\pos_d$ with a sampling step $d=\Om{\log n}$, we modify our \bwte\
algorithm as follows. At the end of pass~$h$, instead of $\bwtext = \bwt(T_h
\cdots T_1)$ we store on disk the pairs $\langle i_1, j_1\rangle, \langle
i_2, j_2\rangle, \ldots \langle i_k, j_k\rangle$ such that
$\sa_{h}[i_\ell]=j_\ell$ is a multiple of $d$ (here $\sa_{h} =
\sa(T_{h}\cdots T_1)$). These pairs are sorted according to their first
component and essentially represent $\pos_d(T_h \cdots T_1)$. The update of
this set of pairs at pass $h+1$ is straightforward: the second component does
not change, whereas the value $i_\ell$ must be increased by the number of new
suffixes which are lexicographically smaller than $i_\ell$ old suffixes. This
can be done via a sequential scan of the already computed set of pairs and of
the $\gap$ array. Since the set $\pos_d$ contains $n/d = \Ob{n/\log n}$
pairs, we have:

\begin{theorem}\label{teo:posd}
We can compute $\pos_d$ in $\Ob{n/M}$ passes over $\Th{n}$ bits of disk data,
using $n$ bits of disk working space. The total number of I/Os is
$\Ob{n^2/(MB\log n)}$ and the CPU time is $\Ob{n^2/M}$.\qed
\end{theorem}

\vspace{-.3cm}
\section{Lightweight Scan-Based BWT Inversion}
\label{sec:inversion}

\vspace{-.3cm} The standard algorithm for inverting the BWT is based on the
fact that the ``successor'' of character \(\bwt [i]\) in $T$ is
$\bwt[\Psi[i]]$. Since we can set up a pointer from position $i$ to position
$\Psi[i]$ for $i=1,\ldots,n$ in linear time, to retrieve $T$ we essentially
need to solve a {\em list ranking} problem in which we have to restore a
sequence given the first element and a pointer to each element's successor.
The na\"{i}ve algorithm for list ranking --- follow each pointer in turn ---
is optimal when the permuted sequence and its pointers fit in memory, but
very slow when they do not.  List ranking in external memory has been
extensively studied, and Chiang \emph{et al.}~\cite{CGG+95} showed how to
reduce this problem to sorting a set of $n$ items (recursively), each of size
$\Theta(\log n)$ bits. If we invert \bwt\ by turning it into an instance of
the list-ranking problem and solve that by using Chiang \emph{et al.}'s
algorithm, then we end up with a solution requiring \(\Th{n \log n}\) bits of
disk space. We now show that, still using Chiang \emph{et al.}'s algorithm as
a subroutine, we can invert \bwt\ using a sorting primitive now applied on
$O(n/\log n)$ items, for a total of \(\Ob{n}\) bits of disk space. In the
full paper we will also show how we can similarly recover $T$ from the array
$\Psi$ still using \(\Ob{n}\) bits of total disk space, and how to take
advantage of the $\pos_d$ array.

To leave the discussion general we write $\sort(x)$ to indicate the {cost} of
sorting $x$ items without detailing the underlying model of computation. Our
algorithm for BWT-inversion works in \(\Ob{\log n}\) rounds, each working on
two files. The first file contains a set $\Set$ of \(n / \log n\) substrings
of $T$.  Each substring is prefaced by a header, which specifies (i) the
position in \bwt\ of the substring's first character, (ii) the position in
\bwt\ of the successor in $T$ of the substring's last character, (iii)
eventually, the character whose index is in (ii) and, (iv) eventually, the
substring's position in a certain partial order that we will define later.
These substrings are non-overlapping and their length increases as the
algorithm proceeds with its rounds.  The second file contains the \bwt\ plus
an $n$-bit array $\bwtMark$ which marks the characters of \bwt\ already
appended to some substring of $\Set$. The overall space taken by both files
is $O(n)$ bits.

The main idea underlying our algorithm is to cover $T$ by the substrings of
$\Set$, avoiding their overlapping. The substrings of $\Set$ consist
initially of the characters which occupy the first  \(n / \log n\) positions
of \bwt; then, they are extended one character after the other along the
\(\Ob{\log n}\) rounds, always taking care that they do not overlap. If, at
some round, $c$ of those substrings become adjacent in $T$, they are merged
to form one single, longer substring which is then inserted in $\Set$, and
those $c$ constituting substrings are deleted. In each round, we use Chiang \emph{et al.}'s list-ranking algorithm on the headers both to detect when substrings become adjacent and to determine the order in which we should merge adjacent substrings.  Our algorithm preserves the
condition $|\Set| = n/\log n$, by selecting $(c-1)$ new substrings which are
inserted in $\Set$ and consist of one single character not already belonging
to any substring of $\Set$. This is easily done by scanning $\bwt$ and
$\bwtMark$ and taking the first $(c-1)$ characters of \bwt\ which result {\em
unmarked} in $\bwtMark$. Keep in mind that whenever a character is appended
to a substring, its corresponding bit in $\bwtMark$ is set to~$1$.

\begin{theorem} \label{teo:bwtInv}
We can invert the BWT in $O(n/M)$ passes on one disk and $O(\log^2 n)$ passes
on two or more disks. If we allow random (i.e. non sequential) disk accesses
we can invert the BWT in  $O(\frac{n}{B} \log_{M/B} \frac{n}{B\log n})$ I/Os.
For all algorithms the total disk usage is $\Th{n}$ bits.
\end{theorem}

\begin{proof}
A round of the algorithm is implemented as follows. We sort
the substrings according to their headers' second components (i.e. positions
in \bwt\ of their following character in $T$), extract the headers into a
separate file, and then fill in their third components.  To do this with one
disk, we need at most $\Ob{n / (M \log n)}$ passes over the headers and \bwt;
with more than one disk, we can do it with $\Ob{1}$ passes.  Whenever we
reach a character in \bwt\ which is pointed to by some header (i.e. follows
the substring of $\Set$ corresponding to that header), then we copy this
character in the (third component of the) header and then update its second
component to make it point to the position in \bwt\ of the new character's
successor in $T$.  (This is done by keeping track of distinct
characters' frequencies in \bwt\ as we go.)  When we are finished filling in
the third components, we append the new (single) characters to the substrings
of $\Set$, update \bwtMark\ by marking $\bwtMark[i]$ if the character in
position $i$ has just been appended, and reinsert the headers.  All these
steps have total cost \(\Ob{\sort(n / \log n)}\).

The two major difficulties we face are, first, that the starting position of
a substring in \bwt\ does not tell anything about its position in $T$;
second, that the first \(n / \log n\) characters in \bwt\ will usually not be
spread evenly throughout $T$. Therefore, we will eventually need to sort the
substrings into the order in which they appear in $T$; in the meantime, we
need to prevent them overlapping. The first of these problems is easier, and
will help us with the second.  Assume that, after the last round, the
substrings cover $T$ and do not overlap. Because the first character in each
substring is the successor in $T$ of the last character in some other
substring (we consider \(T[1]\) to be \(T[n]\)'s successor), we can sort the
substrings by list ranking, as follows.  We extract the headers and apply
list ranking to their first and second components, which has cost
\(\Ob{\sort(n / \log n)}\).  We store each header's rank as the header's
third component, then reinsert the headers into the substrings.  The headers'
third components now tell us in what order the substrings appear in $T$, so
we can sort the substrings by them to obtain $T$.

The second difficulty we face is in preventing the substrings overlapping
during the rounds: if we simply stop appending to some substrings because the
characters we would  append are already in other substrings, then the number
of characters we append per round decreases and we may use more than
\(\Ob{\log n}\) rounds; on the other hand, if we start new substrings without
reducing the number of old ones, then we may store more than \(\Ob{n / \log
n}\) substrings and so, because each has three \(\Th{\log n}\)-bit pointers,
use more than \(\Ob{n}\) bits of disk space.   Our solution is to sort the
substrings by list ranking (as described above) during every round, to find
maximal sequences of adjacent substrings; we merge adjacent substrings into
one longer substring, which is inserted in $\Set$, and delete the others;
pointers can easily be kept correctly.  Again, these steps take a cost of
\(\Ob{\sort(n / \log n)}\).

At each round,
$n/\log n$ new characters are appended to the substrings of $\Set$. Since
these substrings are guaranteed not to overlap we are guaranteed that $O(\log
n)$ rounds suffice to append all characters in $T$ to $\Set$'s substrings.
The proof follows recalling that each round requires a constant number of
sort/scan primitives over $\Ob{n}$ items  and that, in our model, $\sort(x)$
takes $O(x/M)$ passes on one disk, $O(\log x)$ passes on two or more disks,
or $\Ob{\frac{x}{B} \log_{M/B} \frac{x}{M}}$ non-sequential I/Os.\qed
\end{proof}

\vspace{-.2cm}
\section{Lower bounds} \label{sec:lb}

\vspace{-.3cm} Our scan-based algorithms to compute or invert the \bwt\ have
a product ``memory's size $\times$ number of passes'' which is \(\Ob{n \log
n}\) bits.  We prove in this section that we cannot reduce them to \(\os{n}\)
bits via any algorithm that uses only {\em one single disk} (accessed
sequentially). Hence our algorithms are an \(\Ob{\log n}\)-factor from the
optimal.  We note that our lower bound is best-possible because, if we have
$\Omega (n)$ bits of memory, then we can read the input into internal memory
with one pass over the disk and then compute the BWT there using, e.g.,
Theorem~\ref{teo:bwt}.

In a recent paper~\cite{Gag09} we observed that, if the repeated substring is
larger than the product of the size of the memory and the number of passes,
then an algorithm that uses multiple passes but only one disk still cannot
take full advantage of the string's periodicity.  Using properties of De
Bruijn sequences we proved that, with polylogarithmic memory and
polylogarithmic passes over one disk, we cannot achieve entropy-only bounds
and, therefore, we also cannot compute the BWT.  In that paper, however, we
were mostly concerned with low-entropy bounds, and only considered the BWT as
a means to achieve them. Our new lower bound for the BWT alone is stronger,
with a simple and direct proof. Our previous lower bound was based on a
technical lemma that we can restate as follows:
\begin{lemma} \label{lem:function}
Consider an invertible function from strings to strings and a machine that
computes (or inverts) that function using only one disk.  We can compute any
substring of an input string given 1) for each pass, the machine's memory
configurations when it reaches and leaves the part of the disk that initially
(resp., eventually) holds that substring, and 2) the eventual (resp.,
initial) contents of that part of the disk.\qed
\end{lemma}
Our new lower bound is based on the same lemma but, instead of combining it
with properties of De Bruijn sequences, we now combine it with a property of
the BWT itself, demonstrated by Mantaci, Restivo and Sciortino~\cite{MRS03}:
it turns periodic strings with relatively short periods into strings
consisting of relatively few runs.
\begin{lemma} \label{lem:runs}
If $T$ is periodic and its minimum period $r$ divides $n$, then \(\bwt (T)\)
consists of $r$ runs, each of length \(n / r\) and containing only one
distinct character.
\end{lemma}

Lemma~\ref{lem:function} implies that, if the initial contents of some part
of the disk are much more complex than its eventual contents (or vice versa),
then the product of the memory's size and the number of passes must be at
least linear in the initial (resp., eventual) contents' complexity. To see
why, consider that we can compute the initial contents from the eventual
contents (or vice versa) and two memory configurations for each pass;
therefore, the product of the memory's size and the number of passes must be
at least the difference between the complexities. Lemma~\ref{lem:runs}
implies that, if $T$ is periodic, then short substrings of \(\bwt (T)\) are
simple. Combining these ideas in a fairly obvious way gives us our lower
bound.
\begin{theorem} \label{teo:LBs}
In the worst case, we can neither compute nor invert the BWT using only one
disk when the product of the memory's size in bits and the number of passes
is $\os{n}$.
\end{theorem}

\begin{proof}
Suppose $T$ is
periodic with minimum period $r$, where $r$ is sublinear in $n$ but still
larger than the product of the memory's size and the number of passes, and
consider any algorithm $A$ that computes \(\bwt (T)\) using only one disk.
Without loss of generality, we can assume $A$ completely overwrites $T$.
Therefore, by Lemma~\ref{lem:runs}, $A$ replaces each copy of $T$'s repeated
substring, which has length $r$, by a substring of \(\bwt (T)\) consisting of
runs of length \(n / r\) (except possibly for the first and last).  Notice
each new substring consists of at most \(r / (n / r) + 1 = \os{r}\) runs, so
we can store it in $\os{r}$ bits, whereas storing $T$'s repeated substring
takes \(\Omega (r)\) bits in the worst case. Lemma~\ref{lem:function} says we
can compute $T$'s repeated substring from one of these new substrings and two
memory configurations for each pass $A$ makes; it follows that two times the
memory's size times the number of passes must be \(\Omega (r)\) bits in the
worst case. Since $r$ can be any integer-valued function in $\os{n}$, it
follows that we cannot compute the BWT using only one disk when the product
of the memory's size and the number of passes is $\os{n}$.

Now consider any algorithm $A'$ that inverts \(\bwt (T)\) using only one
disk.  Again without loss of generality, we can assume $A'$ completely
overwrites \(\bwt (T)\).  Therefore, by Lemma~\ref{lem:runs}, with each copy
of $T$'s repeated substring, $A$ replaces on the disk a substring of \(\bwt
(T)\) consisting of $\os{r}$ runs.  It follows, by the same arguments as
above, that we cannot invert the BWT using only one disk when the product of
the memory's size and the number of passes is $\os{n}$. \qed
\end{proof}

\bibliographystyle{plain}
\bibliography{compression,webIR}

\end{document}